\documentclass[journal, 10pt]{IEEEtran}
\IEEEoverridecommandlockouts
\usepackage{epsfig}
\usepackage{cite}
\usepackage{amsfonts}
\usepackage{graphics}
\usepackage{amsmath}
\usepackage{amssymb}
\usepackage{amsthm}
\usepackage{latexsym}
\usepackage[]{fontenc}
\usepackage{psfrag}
\usepackage{color}
\usepackage{dsfont}
\usepackage{booktabs}
\usepackage{bm}
\usepackage{mathabx}

\newtheorem{proposition}{Proposition}

\def\ba{{\boldsymbol{a}}}
\def\bb{{\boldsymbol{b}}}

\def\bg{{\boldsymbol{g}}}
\def\bh{{\boldsymbol{h}}}

\def\bn{{\boldsymbol{n}}}

\def\bw{{\boldsymbol{w}}}
\def\bx{{\boldsymbol{x}}}
\def\by{{\boldsymbol{y}}}

\def\bA{{\boldsymbol{A}}}

\def\bD{{\boldsymbol{D}}}

\def\bG{{\boldsymbol{G}}}
\def\bH{{\boldsymbol{H}}}
\def\bI{{\boldsymbol{I}}}

\def\bN{{\boldsymbol{N}}}

\def\bY{{\boldsymbol{Y}}}

\begin{document}

\title{Performance Analysis of NOMA in Training Based Multiuser MIMO Systems}

\author{\IEEEauthorblockN{Hei Victor Cheng, Emil Bj\"ornson, and Erik G. Larsson}

\IEEEauthorblockA{Department of Electrical Engineering (ISY), Link\"oping University, Sweden\\
Email: \{hei.cheng, emil.bjornson, erik.g.larsson\}@liu.se}

\thanks{This work was supported by the Swedish Research Council (VR), the Link\"oping University Center for Industrial Information Technology (CENIIT), and the ELLIIT. Part of this work has been presented at IEEE International Workshop on Signal Processing Advances in Wireless Communications (SPAWC) 2017 \cite{CBL2017}, however there is an error in the Fig. 2 which is corrected in this paper.}}

\maketitle

\begin{abstract}
This paper considers the use of NOMA in multiuser MIMO systems in practical scenarios where CSI is acquired through pilot signaling. A new NOMA scheme that uses shared pilots is proposed. Achievable rate analysis is carried out for different pilot signaling schemes including both uplink and downlink pilots. The achievable rate performance of the proposed NOMA scheme with shared pilot within each group is compared with the traditional orthogonal access scheme with orthogonal pilots. Our proposed scheme is a generalization of the orthogonal scheme, and can be reduced to the orthogonal scheme when appropriate power allocation parameters are chosen. Numerical results show that when downlink CSI is available at the users, our proposed NOMA scheme outperforms orthogonal schemes. However with more groups of users present in the cell, it is preferable to use multi-user beamforming in stead of NOMA.
\end{abstract}

\section{Introduction}

Non-orthogonal-multiple-access (NOMA) is a new multiple-access concept
proposed for next generation wireless networks \cite{SKB2013}.  The key
idea behind NOMA is the use of superposition coding \cite{Cover1972},
and associated interference cancellation techniques, to serve multiple
terminals in the same time-frequency slot. This is classified as NOMA in the power domain.
NOMA provides the ability to increase capacity, especially when the number of users exceeds the
dimension of the channel coherence interval, or the number of spatial
dimensions (antennas) available for multiplexing is limited.  The
technology is currently attracting much attention \cite{DWYHIW2015,DAP2016,DLCSEIP2017,NOMAbook}. In the standardization of 3GPP-LTE-Advanced
networks, a NOMA technique for the downlink (DL), called multiuser
superposition transmission (MUST), was recently proposed \cite{MUST}.

Concurrently, multiuser MIMO is becoming a cornerstone technology in emerging standards for
wireless access. The idea is to use multiple, phase-coherently
operating antennas at the base station to simultaneously serve many
terminals and separate them in the spatial domain. The basic multiuser MIMO concepts and the associated
information theory go back a long time \cite{Swales1990,Anderson1991,CS2003}.  The
most useful form of multiuser MIMO is massive MIMO, which emerged more
recently \cite{TM2010,MLYN2016}. In massive MIMO, the base
stations use hundreds of antennas to serve tens of terminals --
harnessing a large spatial multiplexing gain for high area
throughputs, as well as a large array gain for improved coverage.

The question addressed in this paper is under what circumstances the
use of NOMA can provide gains in multiuser MIMO systems. While this question per se is not new, no existing study to the
authors' knowledge addressed it under realistic assumptions on the
availability of channel state information (CSI).  Specifically,
previous work either assumed perfect CSI \cite{DAP2016,DYFP2014} or
only statistical CSI \cite{SHIP2015}. In contrast, we consider the use
of training (pilot transmission) to acquire estimated CSI, and we
derive rigorous capacity bounds for NOMA-based access under these practical conditions. Training-based NOMA schemes have been considered in \cite{YDFK2016}, but only for single-antenna systems and hence only downlink pilots are sent to the users for estimating their effective channel gains. Moreoever, in multi-user MIMO the effective channels depend on the beamforming, which complicates
the analysis. Beamforming with imperfect CSI also creates extra interference to the users, which has not been investigated in the literature. In contrast, in this work pilots are transmitted on the uplink (UL), facilitating the base station to estimate all channels. By virtue of reciprocity and time-division-duplex (TDD) operation, the so-obtained estimates constitute legitimate estimates of the downlink channel as well and can be used for coherent beamforming. However, since the terminals do not know their effective channels, we consider also the possibility of sending (beamformed) pilots in the DL.

The assumptions made on availability of CSI are critical
in the analysis of wireless access performance: Perfect CSI (or even
high-quality CSI) is unobtainable in environments with mobility, and performance analyses conducted under perfect-CSI
assumptions often yield significantly overoptimistic results. Conversely,
the reliance on only statistical CSI precludes the full exploitation
of spatial multiplexing gains, rendering any performance results
overpessimistic. The quality of the channel estimates that can
ultimately be obtained is dictated by the length of the channel
coherence interval (CI) (product of the coherence time and the coherence
bandwidth): the higher mobility, the less room for pilots, the
lower-quality CSI -- and vice versa. Since the coherence time is
proportional to the wavelength, the use of higher carrier-frequencies
accentuates this problem. In high mobility and at high frequencies,
the channel coherence may become very short and eventually one is
forced to use non-coherent communication techniques \cite{JCPM2016}.

The specific technical contributions of this paper are:
\begin{itemize}
\item We propose a training scheme to obtain CSI and utilize the NOMA concept in a DL multiuser MIMO system.
\item The derivation of new, rigorous, semi-closed form lower bounds
  on the DL capacity in multiuser MIMO with NOMA, with and without DL pilots.
\item A numerical demonstration that NOMA can give gains in multiuser
  MIMO with estimated CSI under appropriate conditions, and a discussion of relevant application
  scenarios, most importantly that of rate-splitting and multicasting.
\end{itemize}

\nopagebreak

\section{System Model}
We consider a single-cell massive MIMO system with $M$ antennas at the base station (BS) and $K$ (even number) single-antenna users. Among these users, $K/2$ of them are located in the cell center, while the other $K/2$ users are at the cell edge. TDD operation is assumed and therefore the BS acquires downlink channel estimates through uplink pilot signaling, by exploiting channel reciprocity. These estimates are used to perform downlink multiuser beamforming. These operations have to be done within the same CI, where the channels are approximately constant. Therefore the more symbols spent on uplink training, the fewer symbols are available for data. We consider non-line-of-sight communication and model the small-scale fading for each user as independent Rayleigh fading. We denote the small-scale fading realizations for the users at the cell center as
\begin{equation}
\bg_k \sim CN(\mathbf{0}, \bI_M), \quad k=1,\ldots, K/2.
\end{equation}
The corresponding large-scale fading parameters are $\beta_g^k>0,~k=1,\ldots,K/2$ for the users in the cell center; the actual channel realization is then $\sqrt{\beta_g^k}\bg_k$.

Similarly, the small-scale fading realizations for the users at the cell edge are denoted as
\begin{equation}
\bh_k \sim CN(\mathbf{0}, \bI_M), \quad k=1,\ldots, K/2.
\end{equation}
The corresponding large-scale fading parameters are $\beta_h^k>0,~k=1,\ldots,K/2$ for the users at the cell edge. The actual channel is then $\sqrt{\beta_g^k}\bg_k$. The large-scale fading is widely different between the two sets of users: $\beta_g^k\gg \beta_h^k$. Note that this is the scenario of interest to us, but the formulas will actually be valid for any values of $\beta_g^k$ and $\beta_h^k$. The names "cell edge" and "cell center" are just descriptive, but should not be interpreted literally.

The BS is assumed to know the deterministic parameters $\beta_g^k$ and $\beta_h^k$. However the small-scale fading realizations are unknown a priori and changing independently from one CI to another CI. To estimate the small-scale fading realizations at the BS, in traditional TDD multiuser MIMO, orthogonal uplink pilots are transmitted from the users in the cell. However, the number of available orthogonal pilot sequences is limited by the size of the CI and this effectively limits the number of users that can be scheduled simultaneously.
In this study, we are interested in the case when $K$ is greater than the number of available pilot sequences. To facilitate discussion and analysis, we assume that there are only $K/2$ orthogonal pilot sequences available. With this assumption, we compare two schemes that make use of the $K/2$ pilot sequences differently.

\subsection{Orthogonal Access Scheme}
The first scheme is the traditional orthogonal access scheme \cite{FWC2005} that schedules $K/2$ users in a fraction $\eta$ of time-frequency resources, and then serve the others in the remaining fraction $1-\eta$ of the resources. To minimize near-far effects, we schedule the $K/2$ users at the cell center in the first fraction $\eta$, followed by the other $K/2$ users at the cell edge in the remaining $1-\eta$ of the resources.
From now on we call this \emph{Scheme-O}.

\subsection{Proposed NOMA Scheme}

The second scheme is a generalization of an existing scheme in the NOMA literature\cite{Kim2013}, which creates $K/2$ groups, each with one user at the cell edge and one at the cell center.

In \cite{Kim2013}, the beamformers are selected based on the channel of the cell center user, but NOMA with superposition coding is applied within each group so that the cell edge user can get a separate data signal. The beamformers can be selected to mitigate the inter-group interference. For example, in \cite{Kim2013}, zero-forcing beamforming is applied to cancel inter-group interference. However, this existing scheme can only provide the user at the cell edge with a small data rate. This is so because the beams are directed to the stronger user in the group, thus the weaker user will not have any beamforming gain and this results in low received power and no interference suppression. Moreover, the existing work is based on the impractical assumption of perfect CSI. In \cite{Choi2015} a two-stage beamformer is proposed where the outer stage aims to cancel the inter-group interference and the inner stage beamformer is optimized to enhance the rate performance for the users within the group. However this approach needs perfect CSI at the BS which is hard to obtain in practice and therefore we do not consider it here.

We propose a generalization of the NOMA scheme from \cite{Kim2013} and devise a way to estimate the channels in practice.
To resolve the pilot-shortage problem, we propose to reuse the same pilot for multiple terminals in the same cell. In particular, the BS allocates the same pilot to the two terminals in a group, where one is in the cell center and one is at the cell edge.\footnote{This scheme can be extended to more than two users, and we will briefly discuss about this in Section \ref{practical_issues}.} Since the two users are using the same pilot and have the same small-scale fading statistics, we will later see that the BS cannot distinguish their channel responses. However, the BS can estimate a linear combination of the channels to both terminals from the pilot transmission. This estimate provides a useful description of the combined channel, particularly, if power control is used to even out the pilot signal strengths of the two terminals. In our proposed scheme, the BS  beamforms a combination of the data symbols intended for the two terminals using the estimated channels. We make use of the NOMA concept for which the symbols intended for different users are super-imposed using super-position coding. The cell edge user performs the decoding by treating inter-user interference as noise, while the cell center user decodes the other user's data first and performs interference cancellation before decoding its own data. Since the beamformers are based on the channels of all users, the proposed scheme can deliver good data rates to everyone.

From now on we call this generalized NOMA scheme \emph{Scheme-N}.

Fig. \ref{CI} shows the frame structure for the two schemes, Scheme-O and Scheme-N, and Fig. \ref{superposition} shows the training and beamforming operations for Scheme-N.
\begin{figure}
        \begin{psfrags}
        \psfrag{T}[Bc][Bc][0.8]{One CI}
        \psfrag{T2}[Bc][Bc][0.8]{Another CI}
        \psfrag{edge}[Bc][Bc][0.8]{Cell Edge Users}
        \psfrag{data}[Bc][Bc][0.8]{DL Data:}
        \psfrag{center}[Bc][Bc][0.8]{Cell Center Users}
        \psfrag{p}[Bc][Bc][0.8]{Pilots}
        \psfrag{dp}[Bc][Bc][0.8]{Pilots}
        \psfrag{ul}[Bc][Bc][0.8]{UL}
        \psfrag{dl}[Bc][Bc][0.8]{DL}
        \psfrag{all}[Bc][Bc][0.8]{All Users}
        \center \includegraphics[width=\linewidth,keepaspectratio]{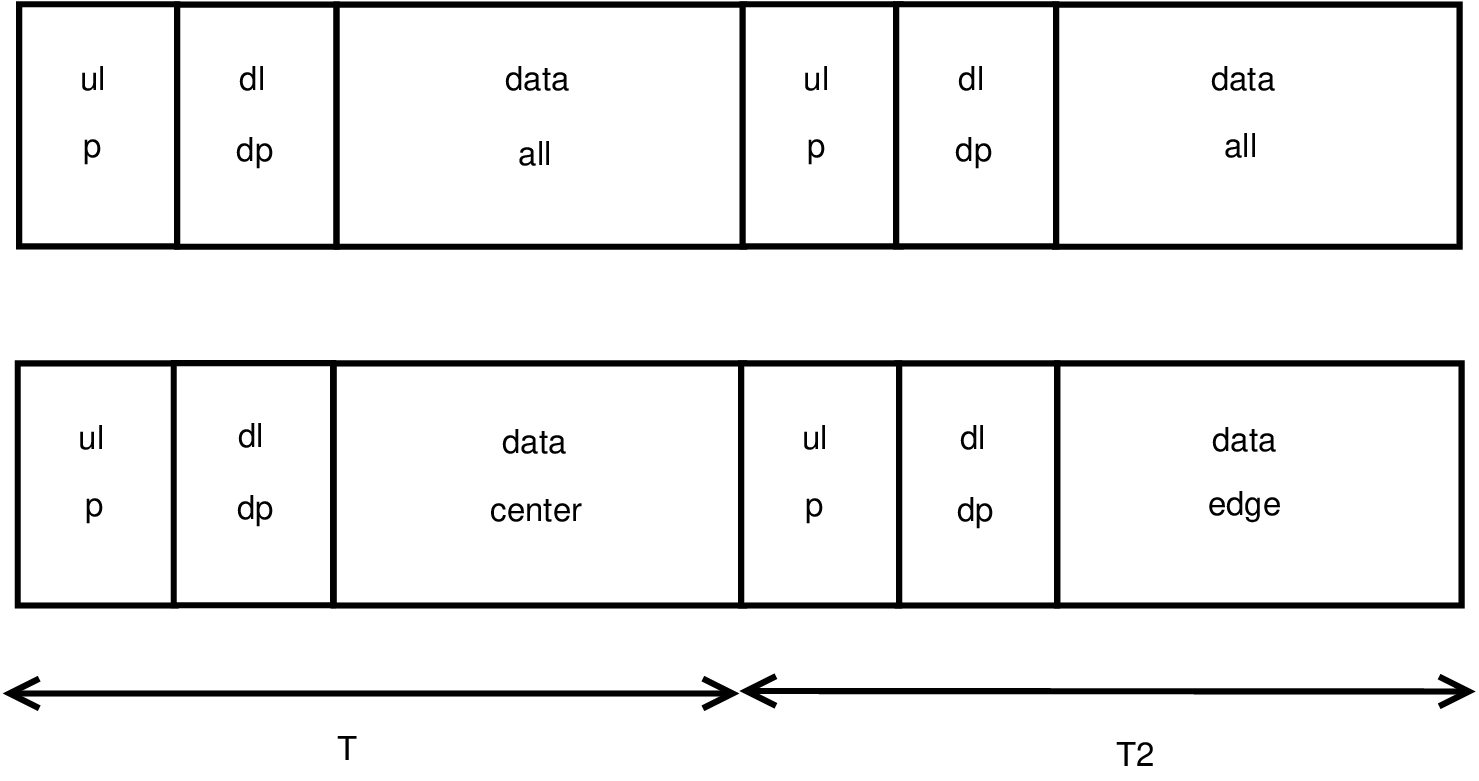}
        \end{psfrags}
        \caption {\label{CI} Frame structure in the considered training based multiuser MIMO systems. Upper figure: frame structure for the proposed Scheme-N, where all users are scheduled by sharing pilots. Bottom figure: common frame structure for Scheme-O, where users are scheduled in different CIs.}
\end{figure}

\begin{figure}
        \begin{psfrags}
        \psfrag{bg}[Bc][Bc][0.8]{user (1,g)}
        \psfrag{bh}[Bc][Bc][0.8]{user (1,h)}
        \psfrag{bg2}[Bc][Bc][0.8]{user (2,g)}
        \psfrag{bh2}[Bc][Bc][0.8]{user (2,h)}
        \psfrag{a}[Bc][Bc][0.8]{(a)}
        \psfrag{b}[Bc][Bc][0.8]{(b)}
        \center \includegraphics[width=\linewidth,keepaspectratio]{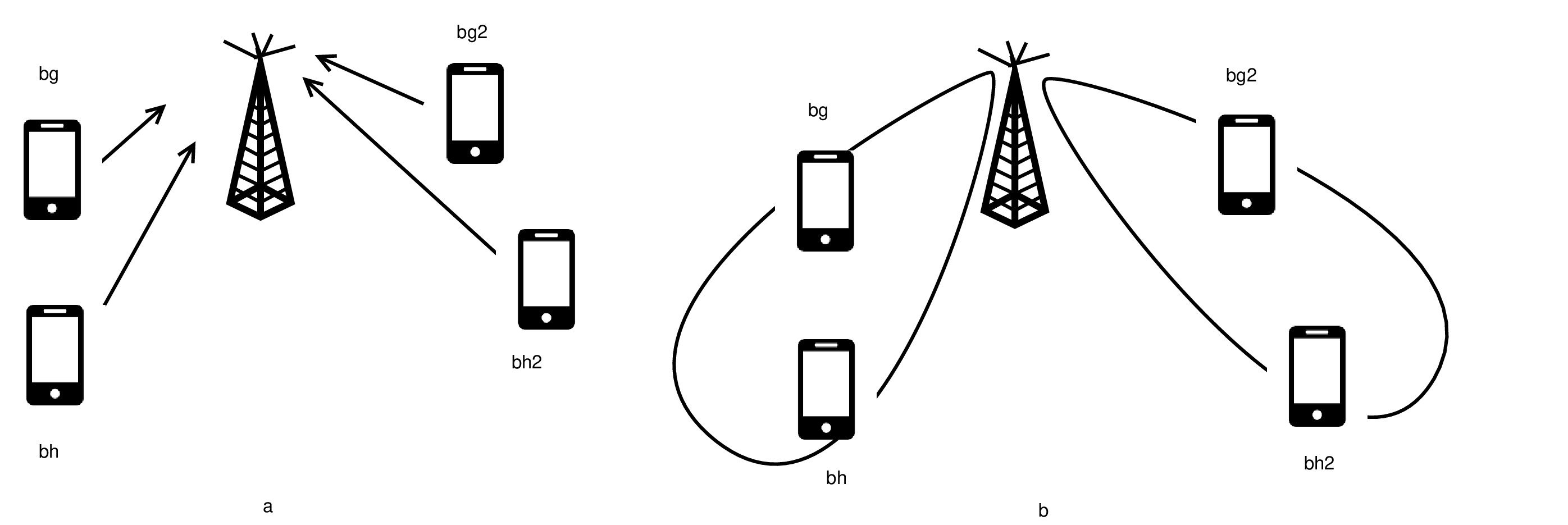}
        \end{psfrags}
        \caption {\label{superposition} The training and the beamforming stages for Scheme-N. (a) the transmissions during the uplink training stage where two users with the same index share the same orthogonal pilot. (b) the beamforming transmission for the data where the same beam is formed for every two users.}
\end{figure}

\section{Uplink Channel Estimation}\label{channel_estimation}
In this section we consider the uplink channel estimation for the two schemes that we are comparing. The channel estimation is different from in conventional systems since the number of users scheduled in one slot and the pilot orthogonality are different. We denote the pilot matrix by $\mathbf{\Phi}\in \mathbb{C}^{K/2\times K/2}$ that contains the $K/2$ orthogonal pilot sequences in its rows, i.e. $\mathbf{\Phi}\mathbf{\Phi}^H=\mathbf{I}_{K/2}$.

For Scheme-O, the $K/2$ users at the cell center are scheduled first, and the received uplink pilot signal $\bY_{uc}^O \in \mathbb{C}^{M\times K/2}$ is
\begin{equation}\label{uloc}
\bY_{uc}^O=\sqrt{p_u}\bG\bD_g\mathbf{\Phi}+\bN_{uc},
\end{equation}
where $\bD_g$ is a diagonal matrix with $\sqrt{\beta_g^1},\ldots,\sqrt{\beta_g^{K/2}}$ on its diagonal.
Then the $K/2$ users at the cell edge are scheduled in a subsequent CI and the received uplink pilot signal $\bY_{ue}^O \in \mathbb{C}^{M\times K/2}$ is
\begin{equation}\label{uloe}
\bY_{ue}^O=\sqrt{p_u}\bH\bD_h\mathbf{\Phi}+\bN_{ue},
\end{equation}
where $\bD_h$ is a diagonal matrix with $\sqrt{\beta_h^1},\ldots,\sqrt{\beta_h^{K/2}}$ on its diagonal.
For Scheme-N, the received uplink pilot signal $\bY_{u}^N \in \mathbb{C}^{M\times K/2}$ is
\begin{equation}\label{uln}
\bY_{u}^N=\sqrt{p_u}\bG\bD_g\bA_g\mathbf{\Phi}+\sqrt{p_u}\bH\bD_h\bA_h\mathbf{\Phi}+\bN_{u},
\end{equation}
where $\bA_g$ and $\bA_h$ are diagonal matrices with $\sqrt{\alpha_g^{1}},\ldots,\sqrt{\alpha_g^{K/2}}$ and $\sqrt{\alpha_h^1},\ldots , \sqrt{\alpha_h^{K/2}}$ on the diagonal respectively. $\bN_{uc}$, $\bN_{ue}$ and $\bN_{u}$ represent the additive noise during pilot transmission with independent and identically distributed (i.i.d.) $CN(0,1)$ entries. $\alpha_h^k\leq1$ and $\alpha_g^k\leq1$ are the positive power control parameters applied to the pilot to (potentially) even out the channel estimation quality between the users in the same group.

Without loss of generality, the $k^{th}$ user at the cell center is paired with the $k^{th}$ user at the cell edge to form the $k^{th}$ group in Scheme-N, and they are using the same pilot sequence. From now on we call the cell edge user in the $k^{th}$ group ``user $(k,h)$'' and the cell center user in the $k^{th}$ group ``user $(k,g)$''. 

\subsection{MMSE Channel Estimation for Scheme-O}

In this subsection, we consider the channel estimation for Scheme-O. The estimates will be used in the next section for performance analysis. The BS first processes the received pilots signals by multiplying with $\mathbf{\Phi}^H$ from the right. The processed pilot signal in \eqref{uloc} becomes
\begin{equation}
\bar{\by}_{uc,k}^O=[\bY_{uc}^O\mathbf{\Phi}^H]_k=\sqrt{p_u\beta_g^k}\bg_k+\bar{\bn}_{uc,k}, \quad k=1,\ldots, K/2,
\end{equation}
where $\bar{\bn}_{uc,k}=[\bN_{uc}\mathbf{\Phi}^H]_k\sim CN(\mathbf{0},\bI_M)$, for the users at the cell center, and where $[\cdot]_k$ denotes the kth column of a matrix . The processed pilot signal in \eqref{uloe} becomes
\begin{equation}
\bar{\by}_{ue,k}^O=[\bY_{ue}^O\mathbf{\Phi}^H]_k=\sqrt{p_u\beta_h^k}\bh_k+\bar{\bn}_{ue,k}, \quad k=1,\ldots, K/2,
\end{equation}
where $\bar{\bn}_{ue,k}=[\bN_{ue}\mathbf{\Phi}^H]_k\sim CN(\mathbf{0},\bI_M)$, for the users at the cell edge.

Based on the processed received pilots, the BS then performs channel estimation. We consider MMSE channel estimation here. Using classical results from \cite{Kay1993}, we obtain the MMSE channel estimate of $\bg_k$ is
\begin{equation}
\hat{\bg}_k=\frac{\sqrt{p_u\beta_g^k}}{p_u\beta_g^k+1}\bar{\by}_{uc,k},\quad k=1,\ldots,K/2
\end{equation}
for users at the cell center and the MMSE estimate of $\bh_k$ is
\begin{equation}
\hat{\bh}_k=\frac{\sqrt{p_u\beta_h^k}}{p_u\beta_h^k+1}\bar{\by}_{ue,k},\quad k=1,\ldots,K/2
\end{equation}
for users at the cell edge.
\subsection{MMSE Channel Estimation for Scheme-N}
Similar to the case of Scheme-O, the BS first processes the received pilot signal by multiplying with $\mathbf{\Phi}^H$ from the right in \eqref{uln} and obtains the processed received signals
\begin{equation}
\begin{aligned}
\bar{\by}_{u,k}^N&=[\bY_{u}^N\mathbf{\Phi}^H]_k=\sqrt{p_u\alpha_g^k\beta_g^k}\bg_k \\
&+\sqrt{p_u\alpha_h^k\beta_h^k}\bh_k+\bar{\bn}_{u,k}, \quad k=1,\ldots, K/2,
\end{aligned}
\end{equation}
where $\bar{\bn}_{u,k}=[\bN_{u}\mathbf{\Phi}^H]_k\sim CN(0,\bI_M)$.
Then the MMSE channel estimate of $\bg_k$ for a user in the cell center is
\begin{equation}
\begin{aligned}
\hat{\bg}_k&=\frac{\sqrt{p_u\alpha_g^k\beta_g^k}}{p_u\alpha_g^k\beta_g^k+p_u\alpha_h^k\beta_h^k+1}\bar{\by}_{u,k}^N, \quad k=1,\ldots,K/2.
\end{aligned}
\end{equation}
The MMSE channel estimate of $\bh_k$ for a user at the cell edge is
\begin{equation}
\hat{\bh}_k=\frac{\sqrt{p_u\alpha_h^k\beta_h^k}}{p_u\alpha_g^k\beta_g^k+p_u\alpha_h^k\beta_h^k+1}\bar{\by}_{u,k}^N, \quad k=1,\ldots,K/2.
\end{equation}

We observe that $\hat{\bg}_k$ and $\hat{\bh}_k$ are parallel, thus the BS cannot distinguish between the channel ``direction'' of users that share the same pilot. This effect is a consequence of pilot contamination. Pilot contamination is a major issue in massive MIMO system, since it makes it hard for the BS from performing coherent beamforming only towards one of the users that share a pilot\cite{TM2010}. In contrast, if the same data is multicasted to multiple users, it is desirable to jointly beamform towards all of them. Pilot contamination is then useful to reduce the pilot overhead\cite{YMA2013}. In this paper, we will show how to exploit NOMA to send different data to the users that share a pilot.

One alternative way to utilize the uplink pilots is to estimate the linear combination $$\bw_k=\sqrt{\alpha_g^k\beta_g^k}\bg_k+\sqrt{\alpha_h^k\beta_h^k}\bh_k$$ of the channels. The MMSE estimate of $\bw_k$ for group $k$ is

\begin{equation}
\hat{\bw}_k=\frac{\sqrt{p_u}\alpha_g^k\beta_g^k+\sqrt{p_u}\alpha_h^k\beta_h^k}{p_u\alpha_g^k\beta_g^k+p_u\alpha_h^k\beta_h^k+1}\bar{\by}_{u,k}^N, \quad k=1,\ldots,K/2.
\end{equation}
Note that $\hat{\bw}_k$ is also parallel with $\hat{\bg}_k$ and $\hat{\bh}_k$. The choice of channel estimate does not matter because in either case the channel estimates are linearly scaled versions of the processed pilot signal. Hence the beamforming directions suggested by the estimates are the same by using any one of the estimators. Since we need to normalize the beamformer to satisfy the power constraint, the scaling disappears after normalization and therefore does not affect the rate.

\subsection{Interference-Limited Scenarios}
We can obtain a special case by assuming there is no noise during the uplink training, or equivalently that the uplink power $p_u$ goes to infinity. This yields as an upper bound on the performance of all the schemes. It is also a good approximation of the interference-limited scenario with high SNR, but large inter-user interference.

For Scheme-O, noise-free channel estimation implies that the channels are perfectly known at the BS, due to the fact that all users use orthogonal pilots in the uplink training, i.e.,
\begin{equation}
\hat{\bg}_k=\bg_k,\quad k=1,\ldots,K/2,
\end{equation}
and
\begin{equation}
\hat{\bh}_k=\bh_k,\quad k=1,\ldots,K/2.
\end{equation}
In contrast, for Scheme-N, the channel estimate at the BS will still be a linear combination of the channels because of the use of the same pilot in each group. The noise-free estimate of $\bw_k$ becomes

\begin{equation}
\hat{\bw}_k=\sqrt{\alpha_h^k\beta_h^k} \bh_k+\sqrt{\alpha_g^k\beta_g^k}\bg_k=\bw_k, \quad k=1,\ldots,K/2.
\end{equation}

\section{Performance Analysis}
In this section, we analyze the ergodic achievable rates of Scheme-O and Scheme-N under imperfect channel estimation. In wireless systems with fast fading channels, channel codes span many realizations of the fading process. Therefore the ergodic achievable rate is an appropriate metric to characterize the performance of coded systems in fast fading environment. It is commonly adopted in the multiuser MIMO literature, especially when the number of antennas is large. We make use of the UL channel estimates from Section \ref{channel_estimation} for downlink beamforming, by assuming perfect reciprocity between UL and DL. The channel estimation errors are taken into account in the ergodic achievable rate expressions. We separate the analysis into three parts, namely the cases with and without instantaneous DL CSI, and the case with estimated DL channel gains. The case with instantaneous downlink CSI is unobtainable in practice, and used only as a benchmark.

Note that the effective ergodic rate have a prelog penalty $\left(1-\frac{K}{2T}\right)$ for the case without DL pilots, where $T$ is the size of the CI. This penalty accounts for the loss from spending $\frac{K}{2T}$ of every CI to estimate the channels. For the case with DL pilots, the pre-log penalty is $\left(1-\frac{K}{T}\right)$.

\subsection{Downlink Signal Model}
Denote by $p_d$ the DL transmission power normalized by the noise variance. For Scheme-O, the received signal for user $k$ in the cell center is
\begin{equation}
y_{c,k}=\sqrt{p_d\beta_g^k}\bg_k^T \bx_g+n_{c,k}, \quad k=1,\ldots,K/2,
\end{equation}
and the received signal for user $k$ at the cell edge is

\begin{equation}
y_{e,k}=\sqrt{p_d\beta_h^k}\bh_k^T \bx_h+n_{e,k}, \quad k=1,\ldots,K/2,
\end{equation}
where $\bx_g$ ($\bx_h$) is the signal vector containing data for the cell center users (cell edge users), and $n_{c,k}$ ($n_{e,k}$) is the normalized i.i.d. zero mean unit variance complex Gaussian noise at the $k^{\rm{th}}$  user at the cell center (edge). Before transmission, each data symbol is multiplied with a beamforming vector as

\begin{equation}
\bx_g=\sum_{k=1}^{K/2} \bb_k \sqrt{\gamma_{k,g}^O}s_{k,g}
\end{equation}
for the users in the cell center and

\begin{equation}
\bx_h=\sum_{k=1}^{K/2} \ba_k \sqrt{\gamma_{k,h}^O}s_{k,h}
\end{equation}
for the users at cell edge. In the above equations $\gamma_{k,h}$ ($\gamma_{k,g}$) represents the non-negative power control coefficients for user $k$ at the cell edge (cell center), and $s_{k,h}$ ($s_{k,g}$) is the data symbol intended for user $k$ at the cell edge (cell center) which is zero mean and unit variance. The combined signal vectors $\bx_h$ and $\bx_g$ need to satisfy the power constraint $\mathbb{E}[\bx_h^H \bx_h]\leq 1$ and $\mathbb{E}[\bx_g^H \bx_g]\leq 1$.

In this work we focus on maximum ratio transmission (MRT) which is simple to implement and performs close to optimality in low SNR scenarios,
$$\bb_k=\frac{\hat{\bg}_k^{*}}{\sqrt{\mathbb{E}[||\hat{\bg}_k||^2]}}$$ for the cell center users and $$\ba_k=\frac{\hat{\bh}_k^{*}}{\sqrt{\mathbb{E}[||\hat{\bh}_k||^2]}}$$ for the cell edge users. With the normalized beamforming vectors, the power constraint becomes $\sum_{k=1}^{K/2} {\gamma_{k,g}^O}\leq 1$ and $\sum_{k=1}^{K/2} {\gamma_{k,h}^O}\leq 1$.

For Scheme-N, the received downlink signal for users in the cell center is

\begin{equation}
\begin{aligned}
y_{k,g}&=\sqrt{p_d\beta_g^k}\sum_{i=1}^{K/2} \bg_k^T\ba_i \sqrt{\gamma_{i,h}}s_{i,h}\\
&+\sqrt{p_d\beta_g^k}\sum_{i=1}^{K/2} \bg_k^T\bb_i \sqrt{\gamma_{i,g}}s_{i,g} +n_{k}, k=1,\ldots,K/2.
\end{aligned}
\end{equation}

Similarly, the received downlink signal for users at the cell edge can be written as

\begin{equation}
\begin{aligned}
y_{k,h}&=\sqrt{p_d\beta_h^k}\sum_{i=1}^{K/2} \bh_k^T\ba_i \sqrt{\gamma_{i,h}}s_{i,h}\\
&+\sqrt{p_d\beta_h^k}\sum_{i=1}^{K/2} \bh_k^T\bb_i \sqrt{\gamma_{i,g}}s_{i,g} +n_{k}, k=1,\ldots,K/2.
\end{aligned}
\end{equation}

In Scheme-N, where the BS knows only the linear combination of the channels for the users in the same NOMA group, it regards the estimate as the true channel for both users $(k,g)$ and $(k,h)$ since that is the best estimate available. The combined symbols from both terminals in the same group are weighted with the power control coefficients $\sqrt{\gamma_{k,h}}$ and $\sqrt{\gamma_{k,g}}$. The transmitted symbol in the $k^{th}$ NOMA group is hence $\sqrt{\gamma_{k,h}} s_{k,h} +\sqrt{\gamma_{k,g}} s_{k,g}$. Therefore the power constraint is $\sum_k \gamma_{k,h}+\sum_k\gamma_{k,g}\leq 1$. In this case we have the MRT beamforming vector with normalization

\begin{equation}
\ba_k=\bb_k=\frac{\hat{\bw}_k^{*}}{\sqrt{\mathbb{E}[\|\hat{\bw}_k\|^2]}}.
\end{equation}

\subsection{Performance With Perfect CSI at the Users}
In this subsection, we compute the ergodic achievable rate for the two schemes under the assumption that the DL pilots make perfect DL CSI available at the users. This assumes that DL pilots are sent in each CI and users perform channel estimation to obtain their own channel gain coefficients and the cross-channel gains between different users. The achievable rate is obtained by averaging over all sources of randomness in the channel and noise.

For Scheme-O, every user decodes its own data symbol by treating interference as noise. Since perfect CSI is available, an ergodic achievable rate of user $k$ with beamforming vector $\ba_1,\ldots, \ba_K$ and $\bb_1,\ldots, \bb_K$ can be computed using \cite[Section 2.3.5]{MLYN2016}

\begin{equation}\label{ratepoc}
R_{c,k}^{O}=\left(1-\frac{K}{T}\right)\eta\mathbb{E}\left[\log_2\left(1+\frac{p_d \beta_g^k\gamma_{k,g}^O|\bg_k^T\bb_k|^2}{p_d\beta_g^k\sum_j\gamma_{j,g}^O|\bg_k^T\bb_j|^2+1}\right)\right]
\end{equation}
for the users in the cell center and
\begin{equation}\label{ratepoe}
R_{e,k}^{O}=\left(1-\frac{K}{T}\right)(1-\eta)\mathbb{E}\left[\log_2\left(1+\frac{p_d \beta_h^k\gamma_{k,h}^O|\bh_k^T\ba_k|^2}{p_d\beta_h^k\sum_j\gamma_{j,h}^O|\bh_k^T\ba_j|^2+1}\right)\right]
\end{equation}
for the users at the cell edge.

The ergodic achievable rates are measured in b/s/Hz, and they can be achieved by using Gaussian signaling and codewords that span over all channel realizations. The pre-log factors account for the loss in achievable rate due to the fact that each user is only scheduled for a fraction of the CIs, in time or frequency.

For Scheme-N, recall that we name the $k^{th}$ user at the cell edge as $(k,h)$ and the $k^{th}$ user at the cell center as $(k,g)$. The instantaneous SINR of $s_{k,h}$ of user $(k,g)$ is

\begin{equation}
\mathrm{SINR}_{k,g}=\frac{p_d \beta_g^k\gamma_{k,h}|\bg_k^T\ba_k|^2}{p_d\beta_g^k\sum_{j\neq k}\gamma_{j,h}|\bg_k^T\ba_j|^2+p_d\beta_g^k\sum_{j}\gamma_{j,g}|\bg_k^T\bb_j|^2+1}
\end{equation}
and similarly the instantaneous SINR of $s_{k,h}$ at user $(k,h)$ can be written as
\begin{equation}
\mathrm{SINR}_{k,h}=\frac{p_d \beta_h^k\gamma_{k,h}|\bh_k^T\ba_k|^2}{p_d\beta_h^k\sum_{j\neq k}\gamma_{j,h}|\bh_k^T\ba_j|^2+p_d\beta_h^k\sum_{j}\gamma_{j,g}|\bh_k^T\bb_j|^2+1}.
\end{equation}

The condition that user $(k,g)$ can decode the data intended for user $(k,h)$ is that the ergodic achievable rate of $s_{k,h}$ at user $(k,g)$ is no less than the ergodic achievable rate of $s_{k,h}$ at user $(k,h)$, which is explicitly

\begin{equation}
\mathbb{E}[\log_2(1+\mathrm{SINR}_{k,g})]\geq \mathbb{E}[\log_2(1+\mathrm{SINR}_{k,h})].
\end{equation}

When this condition does not hold, we need to lower the data rate to user $(k,h)$ such that it can be decoded at user $(k,g)$. This can be done by choosing
\begin{equation}
R_{k,h}^{NP}=\min \left(\mathbb{E}[\log_2(1+\mathrm{SINR}_{k,g})], \mathbb{E}[\log_2(1+\mathrm{SINR}_{k,h})]\right).
\end{equation}
Since $\mathbb{E}[\log_2(1+\mathrm{SINR}_{k,h})]$ is an achievable rate for user $(k,h)$, from an information-theoretic perspective any rate that is lower than that is also achievable. Therefore by transmitting with the chosen $R_{k,h}^{NP}$ both users are able to decode the data.

In practice, for (28) to hold we just need to properly control the pilot powers such that (28) holds. Then user $(k,g)$ gathers all received signals over all channel realizations (coherence intervals) and decodes the data for user $(k,h)$. Notice that the SIC is done after the whole codeword is decoded, and not performed in every CI. Therefore it is not a problem if the instantaneous SINR is lower at user $(k,g)$, as long as (28) is satisfied in the long term.

In the typical scenarios of $\beta_h^k \ll \beta_g^k$, there is a wide range of possible choices of power control parameters on the pilots available to satisfy (28). With any choice of power control satisfying (28) we transmit with the super-position coding scheme such that user $(k,h)$ decodes the signal $s_{k,h}$ from $y_{k,h}$ by treating the signal from user $(k,g)$ as noise. Then user $(k,g)$ performs successive interference cancellation such that it first decodes $s_{k,h}$ from $y_{k,g}$ and then subtracts $\sqrt{p_d\beta_g^k}\bg_k^T\ba_k\sqrt{\gamma_{k,h}}s_{k,h}$ from $y_{k,g}$ and decodes $s_{k,g}$ afterwards.

With the superposition coding scheme, the achievable rate of user $(k,g)$ is given in \eqref{ratepnc}
\begin{figure*}
\begin{equation}\label{ratepnc}
R_{k,g}^{NP}=\left(1-\frac{K}{T}\right)\mathbb{E}\left[\log_2\left(1+\frac{p_d \beta_g^k\gamma_{k,h}|\bg_k^T\ba_k|^2}{p_d\beta_g^k\sum_{j\neq k}\gamma_{j,h}|\bg_k^T\ba_j|^2+p_d\beta_g^k\sum_{j\neq k}\gamma_{j,g}|\bg_k^T\bb_j|^2+1}\right)\right]
\end{equation}
\hrulefill
\end{figure*}
and the achievable rate of user $(k,h)$ is given in \eqref{ratepne} on top of next page.
\begin{figure*}
\begin{equation}\label{ratepne}
R_{k,h}^{NP}=\left(1-\frac{K}{T}\right)\mathbb{E}\left[\log_2\left(1+\frac{p_d \beta_h^k\gamma_{k,h}|\bh_k^T\ba_k|^2}{p_d\beta_h^k\sum_{j\neq k}\gamma_{j,h}|\bh_k^T\ba_j|^2+p_d\beta_h^k\sum_{j}\gamma_{j,g}|\bh_k^T\bb_j|^2+1}\right)\right]
\end{equation}
\hrulefill
\end{figure*}

It is worth noticing that when $\alpha_g^k=\gamma_{k,g}=0~\forall k$ and $\alpha_h^k=1$, one can obtain $R_{k,h}^{N}=R_{e,k}^{O}$ with $\eta=1$. Similarly when $\alpha_h^k=\gamma_{k,h}=0~, \forall k$ and $\alpha_g^k=1$, one can obtain $R_{k,g}^{N}=R_{c,k}^{O}$ with $\eta=0$. By using time-sharing between these two extremes, we obtain all the ergodic achievable rates that Scheme-O can attain. This shows that Scheme-N is more general than the traditional scheme with orthogonal access.

\subsection{Performance Without Downlink CSI}
In this section we investigate the case when instantaneous DL CSI is not available, however we assume the channel statistics are known by all parties. This corresponds to the case when no DL pilots are sent and serves as a lower bound on the performance of all the schemes with estimated DL channels.
In this case users utilize the long term statistics as the channel gain and decode the signals, that is, they take the statistical average of the effective gain as an estimate of that gain. Then the achievable rate is obtained by gathering all the symbols over different channel realizations and decoding the signal.

Assume the BS uses the estimated CSI for beamforming to all terminals. Since we are considering MRT beamforming, $\ba_k$ and $\bb_k$ are scaled versions of the channel estimate $\hat{\bw}_k$ which is a scaled version of the processed pilots $\bar{\by}_{u,k}^{N}$. Then the beamforming vector is $\ba_k=\bb_k=c_k\bar{\by}_{u,k}^{N*}$ where the normalizing constant $c_k$ that meets the power constraint can be calculated as
\begin{equation}
c_k=\frac{1}{\sqrt{\mathds{E}[\|\bar{\by}_{u,k}^{N}\|^2]}}=\frac{1}{\sqrt{(p_u\alpha_h^k\beta_h^k+p_u\alpha_g^k\beta_g^k+1)M}}.
\end{equation}
Therefore the received signal at user $(k,g)$ is
\begin{equation}
\begin{aligned}
y_{k,g}&=c_k\sqrt{\beta_g^k}\bg_k^T\bar{\by}_{u,k}^{N*}\sqrt{p_d\gamma_{k,h}} s_{k,h} \\
&+c_k\sqrt{\beta_g^k}\bg_k^T\bar{\by}_{u,k}^{N*}\sqrt{p_d\gamma_{k,g}} s_{k,g}+I_{k,g}+ n_{k,g},
\end{aligned}
\end{equation}
where
\begin{equation}
\begin{aligned}
I_{k,g}&=\sqrt{\beta_g^k}\sum_{j\neq k}c_j\bg_k^T\bar{\by}_{u,j}^{N*}\sqrt{p_d\gamma_{j,h}} s_{j,h}\\
&+\sqrt{\beta_h^k}\sum_{j\neq k}c_j\bg_k^T\bar{\by}_{u,j}^{N*}\sqrt{p_d\gamma_{j,g}} s_{j,g}
\end{aligned}
\end{equation}
is the interference from other groups of users.
Similarly, the received signal at user $(k,g)$ is
\begin{equation}
\begin{aligned}
y_{k,h}&=c_k\sqrt{\beta_h^k}\bh_k^T\bar{\by}_{u,k}^{N*}\sqrt{p_d\gamma_{k,h}} s_{k,h}\\
&+c_k\sqrt{\beta_h^k}\bh_k^T\bar{\by}_{u,k}^{N*}\sqrt{p_d\gamma_{k,g}} s_{k,g}+I_{k,h}+ n_{k,h},
\end{aligned}
\end{equation}
where
\begin{equation}
\begin{aligned}
I_{k,h}&=\sqrt{\beta_h^k}\sum_{j\neq k}c_j\bh_k^T\bar{\by}_{u,j}^{N*}\sqrt{p_d\gamma_{j,h}} s_{j,h}\\
&+\sqrt{\beta_h^k}\sum_{j\neq k}c_j\bh_k^T\bar{\by}_{u,j}^{N*}\sqrt{p_d\gamma_{j,g}} s_{j,g}
\end{aligned}
\end{equation}
is the interference from other groups of users.

Now we make use of the channel statistics to write the received signal at terminal $(k,h)$ as
\begin{equation}\label{receivedsignal}
y_{k,h}=\mathbb{E}\left[c_k\sqrt{\beta_h^k}\bh_k^T\bar{\by}_{u,k}^{N*}\sqrt{p_d\gamma_{k,h}}\right] s_{k,h}+z_{k,h}
\end{equation}
where we have introduced the following effective noise term
\begin{equation}
\begin{aligned}
z_{k,h}&=\left(c_k\sqrt{\beta_h^k}\bh_k^T\bar{\by}_{u,k}^{N*}\sqrt{p_d\gamma_{k,h}}-\mathbb{E}\left[c_k\sqrt{\beta_h^k}\bh_k^T\bar{\by}_{u,k}^{N*}\sqrt{p_d\gamma_{k,h}}\right]\right)\\
&s_{k,h}+c_k\sqrt{\beta_h^k}\bh_k^T\bar{\by}_{u,k}^{N*}\sqrt{p_d\gamma_{k,g}} s_{k,g}+I_{k,h}+ n_{k,h}.
\end{aligned}
\end{equation}
It can be easily verified that $z_{k,h}$ is uncorrelated with the signal term in \eqref{receivedsignal}. Therefore \eqref{receivedsignal} can be regarded as an equivalent scalar channel with deterministic known gain and additive uncorrelated noise. Using the fact that additive Gaussian noise is the worst case uncorrelated noise\cite[Section 2.3.2]{MLYN2016}, the following rate is achievable for user $(k,h)$:
\begin{proposition}\label{iprateh}
The following ergodic rate is achievable for user $(k,h)$ with Scheme-N:
\begin{equation}\label{rateimne}
R_{k,h}^{Nip}=\left(1-\frac{K}{2T}\right)\log_2\left(1+\frac{p_d\lambda_{k,h}\beta_h^k\gamma_{k,h}M}{p_d\lambda_{k,h}\beta_h^k\gamma_{k,g}M+p_d\beta_h^k+1}\right),
\end{equation}
where $\lambda_{k,h}$ is defined as
\begin{equation}
\lambda_{k,h}=\frac{p_u\alpha_h^k\beta_h^k}{p_u\alpha_h^k\beta_h^k+p_u\alpha_g^k\beta_g^k+1}.
\end{equation}
\end{proposition}
\begin{proof} The proof is given in Appendix \ref{proofiprateh}.
\end{proof}
We define
\begin{equation}
\lambda_{k,g}=\frac{p_u\alpha_g^k\beta_g^k}{p_u\alpha_h^k\beta_h^k+p_u\alpha_g^k\beta_g^k+1}
\end{equation}
to quantify the channel estimation quality for the following discussion. Under the condition $\alpha_h^k\beta_h^k\leq \alpha_g^k\beta_g^k$, the effective SINR of the signal $s_{k,h}$ at user $(k,g)$ is greater than the effective SINR of the signal $s_{k,h}$ at user $(k,h)$, i.e.,
\begin{equation}
\frac{p_d\lambda_{k,h}\beta_h^k\gamma_{k,h}M}{p_d\lambda_{k,h}\beta_h^k\gamma_{k,g}M+p_d\beta_h^k+1}\leq\frac{p_d\lambda_{k,g}\beta_g^k\gamma_{k,h}M}{p_d\lambda_{k,g}\beta_g^k\gamma_{k,g}M+p_d\beta_g^k+1}.
\end{equation}
Therefore we can use NOMA where user $(k,g)$ decodes data from user $(k,h)$ and then subtracts it from the received signal $y_{k,g}$.
From the sufficient condition $\alpha_h^k\beta_h^k\leq \alpha_g^k\beta_g^k$ we see that it is better to let the user with larger large-scale fading coefficient perform successive interference cancellation as the condition is easier to satisfy.
\color{black}
We have the new received signal
\begin{equation}
\begin{aligned}
\bar{y}_{k,g}&=y_{k,g}-\mathds{E}\left[c_k\sqrt{\beta_g^k}\bg_k^T\bar{\by}_{u,k}^{N*}\right]\sqrt{p_d\gamma_{k,h}}s_{k,h}\\
&=\mathds{E}\left[c_k\sqrt{\beta_g^k}\bg_k^T\bar{\by}_{u,k}^{N*}\right]\sqrt{p_d\gamma_{k,g}}s_{k,g}\\
&+\left(c_k\sqrt{\beta_g^k}\bg_k^T\bar{\by}_{u,k}^{N*}-\mathds{E}\left[c_k\sqrt{\beta_g^k}\bg_k^T\bar{\by}_{u,k}^{N*}\right]\right)\sqrt{p_d\gamma_{k,h}}s_{k,h}\\
&+\left(c\sqrt{\beta_g^k}\bg_k^T\bar{\by}_{u,k}^{N*}-\mathds{E}\left[c_k\sqrt{\beta_g^k}\bg_k^T\bar{\by}_{u,k}^{N*}\right]\right)\sqrt{p_d\gamma_{k,g}} s_{k,g}\\
&+I_{k,g}+n_{k,g}.
\end{aligned}
\end{equation}

We can similarly write the effective noise as
\begin{equation}
\begin{aligned}
z_{k,g}&=\left(c_k\sqrt{\beta_g^k}\bg_k^T\bar{\by}_{u,k}^{N*}-\mathds{E}\left[c_k\sqrt{\beta_g^k}\bg_k^T\bar{\by}_{u,k}^{N*}\right]\right)\sqrt{p_d\gamma_{k,h}}s_{k,h}\\
&+\left(c\sqrt{\beta_g^k}\bg_k^T\bar{\by}_{u,k}^{N*}-\mathds{E}\left[c_k\sqrt{\beta_g^k}\bg_k^T\bar{\by}_{u,k}^{N*}\right]\right)\sqrt{p_d\gamma_{k,g}} s_{k,g}\\
&+I_{k,g}+n_{k,g}.
\end{aligned}
\end{equation}

\begin{proposition}\label{iprateg}
The following ergodic rate is achievable for user $(k,g)$ with Scheme-N:
\begin{equation}\label{rateimnc}
R_{k,g}^{Nip}=\left(1-\frac{\tau}{2T}\right)\log_2\left(1+\frac{\lambda_{k,g}\beta_g^k\gamma_{k,g}M}{p_d\beta_g^k+1}\right).
\end{equation}
\end{proposition}
\begin{proof} The proof is given in Appendix \ref{proofiprateg}.
\end{proof}
From the ergodic rate expressions, we can observe that the signal terms are proportional to $M$, which is the array gain from coherent beamforming. Moreover, we observe that the total interference from other groups of users is a constant that only depends the user's own large-scale fading, but not on the number of antennas or channel estimation quality. Therefore the only parameters that affect the rate are the power control parameters and the uplink channel estimation quality. Adding more groups of users in Scheme-N will only change the amount of power that is allocated to each group, but not the total interference. Each user at the cell edge is affected by coherent interference from the signal intended for the cell center user in its group. However, for the users in the cell center, coherent interference disappears in the successive interference cancellation and the only effect of the pilot contamination is the degraded channel estimation quality.

Using similar calculations, we obtain the ergodic achievable rate expressions for Scheme-O. For users in the cell center, we have
\begin{equation}\label{rateimoc}
\begin{aligned}
R_{c,k}^{Oip}&=\left(1-\frac{K}{2T}\right)\eta\log_2\left(1+\frac{\lambda_{k,g}^O\beta_g^k\gamma_{k,g}^OM}{p_d\beta_g^k+1}\right),\\
& k=1,\ldots, K/2
\end{aligned}
\end{equation}
where
\begin{equation}
\lambda_{k,g}^O=\frac{p_u\beta_g^k}{p_u\beta_g^k+1}.
\end{equation}
For users at the cell edge, we have
\begin{equation}\label{rateimoe}
\begin{aligned}
R_{e,k}^{Oip}&=\left(1-\frac{K}{2T}\right)(1-\eta)\log_2\left(1+\frac{\lambda_{k,h}^O\gamma_{k,h}^O M}{p_d\beta_h^k+1}\right),\\
& k=1,\ldots, K/2
\end{aligned}
\end{equation}
where
\begin{equation}
\lambda_{k,h}^O=\frac{p_u\beta_h^k}{p_u\beta_h^k+1}.
\end{equation}

As in the case with perfect CSI at the users, when we set $\alpha_h^k=\gamma_{k,h}=0, \forall k$ and $\alpha_g^k=1$ in Scheme-N we get the achievable rate of the users in the cell center in Scheme-O with $\eta=1$. Setting $\alpha_h^k=\gamma_{k,h}=0, \forall k$ and $\alpha_h^k=1$ we get the achievable rate of the users at the cell edge in Scheme-O with $\eta=0$. By using time-sharing between these two extremes, we obtain all the ergodic achievable rates that Scheme-O can attain.

\subsection{Performance With Estimated Downlink CSI}
DL CSI does not come for free. In practice some form of estimation of the beamformed channel gain is usually needed. In this subsection, we investigate the performance of the two schemes when we send DL (beamformed) pilots \cite{NL2016} for the channel estimation. For Scheme-O every user receives its own orthogonal pilot. For Scheme-N, since we are using the same beamformer for the pair of users in every group $k$, only one downlink pilot is needed for every pair of users. In this case the users estimate their effective channel gain and perform a form of ``equalization'' using the estimated channel gain (see below for the details).

We denote the channel gain at user $(k,h)$ as $f_{k,h}\triangleq \bh_k^T\ba_k$. Then the received pilot at each of these users is
\begin{equation}
y_{dpk,h}=f_{k,h} \sqrt{p_d\beta_h^k}+ n_{dpk,h}, k=1,\ldots,K/2.
\end{equation}
Assuming LMMSE estimation\cite{Kay1993} at the user, we obtain the estimate
\begin{equation}
\begin{aligned}
\hat{f}_{k,h}&=\mathbb{E}[f_{k,h}]\\
&+\frac{\sqrt{\beta_h^kp_d}\mathrm{Var}[f_{k,h}]}{\beta_h^kp_d\mathrm{Var}[f_{k,h}]+1}\left(y_{dpk,h}-
\sqrt{\beta_h^kp_d}\mathbb{E}[f_{k,h}]\right),
\end{aligned}
\end{equation}
of the channel gain where
\begin{equation}
\begin{aligned}
\mathbb{E}[f_{k,h}]&=\sqrt{M\lambda_{k,h}},\\
\mathrm{Var}[f_{k,h}]&=1.
\end{aligned}
\end{equation}
The estimation quality will improve with $M$ as the mean of the channel gain is increasing with $M$ while the variance is constant.

Similarly, denote the channel gain at user $(k,g)$ as $f_{k,g}\triangleq \bg_k^T\bb_k$. The received pilot at each of these user is
\begin{equation}
\begin{aligned}
y_{dpk,g}&=f_{k,g} \sqrt{p_d\beta_g^k}+ n_{dpk,g}, k=1,\ldots,K/2.
\end{aligned}
\end{equation}
Applying LMMSE estimation yields the estimate
\begin{equation}
\begin{aligned}
\hat{f}_{k,g}&=\mathbb{E}[f_{k,g}]\\
&+\frac{\sqrt{\beta_g^kp_d}\mathrm{Var}[f_{k,g}]}{\beta_g^kp_d\mathrm{Var}[f_{k,g}]+1}\left(y_{dpk,g}-
\sqrt{\beta_g^kp_d}\mathbb{E}[f_{k,g}]\right),
\end{aligned}
\end{equation}
where
\begin{equation}
\begin{aligned}
\mathbb{E}[f_{k,g}]&=\sqrt{M\lambda_{k,g}},\\
\mathrm{Var}[f_{k,g}]&=1.
\end{aligned}
\end{equation}
With these estimates of the channel gains, we first divide the received signal at user $(k,h)$ by the channel estimate. This can be seen as a form of equalization, and ideally the ratio $\frac{f_{k,g}}{\hat{f}_{k,g}}$ is one. Then we use the same method as above to obtain the achievable rate of user $(k,h)$ in \eqref{ratekhdp} on top of next page.
\begin{figure*}
\begin{equation}\label{ratekhdp}
R_{k,h}^{Ndp}=\left(1-\frac{K}{T}\right)\log_2\left(1+\frac{p_d\beta_h^k\gamma_{k,h}\left|\mathbb{E}\left[\frac{f_{k,h}}{\hat{f}_{k,h}}\right]\right|^2}{p_d\beta_h^k\gamma_{k,h} \mathrm{Var}\left[\frac{f_{k,h}}{\hat{f}_{k,h}}\right]+p_d\beta_h^k\gamma_{k,g}\left|\mathbb{E}\left[\frac{f_{k,h}}{\hat{f}_{k,h}}\right]\right|^2+
\mathbb{E}\left[\left|\frac{I_{k,h}}{\hat{f}_{k,h}}\right|^2\right]+\mathbb{E}\left[\left|\frac{1}{\hat{f}_{k,h}}\right|^2\right]}\right)
\end{equation}
\hrulefill
\end{figure*}
Similarly for user $(k,g)$, an achievable rate is given in \eqref{ratekgdp} on top of next page.
\begin{figure*}
\begin{equation}\label{ratekgdp}
R_{k,g}^{Ndp}=\left(1-\frac{K}{T}\right)\log_2\left(1+\frac{p_d\beta_g^k\gamma_{k,g}\left|\mathbb{E}\left[\frac{f_{k,g}}{\hat{f}_{k,g}}\right]\right|^2}{p_d\beta_g^k \mathrm{Var}\left[\frac{f_{k,g}}{\hat{f}_{k,g}}\right]+\mathbb{E}\left[\left|\frac{I_{k,g}}{\hat{f}_{k,g}}\right|^2\right]+\mathbb{E}\left[\left|\frac{1}{\hat{f}_{k,g}}\right|^2\right]}\right)
\end{equation}
\hrulefill
\end{figure*}

For Scheme-O, similar techniques can be applied to obtain the achievable rate for the users in the cell center given in \eqref{ratedlocdp} on top of next page.
\begin{figure*}
\begin{equation}\label{ratedlocdp}
R_{c,k}^{Odp}=\left(1-\frac{K}{T}\right)\eta\log_2\left(1+\frac{p_d\beta_g^k\gamma^O_{k,g}\left|\mathbb{E}\left[\frac{f_{k,g}^O}{\hat{f}_{k,g}^O}\right]\right|^2}{p_d\beta_g^k \mathrm{Var}\left[\frac{f_{k,g}^O}{\hat{f}_{k,g}^O}\right]+\mathbb{E}\left[\left|\frac{I^O_{k,g}}{\hat{f}^O_{k,g}}\right|^2\right]+\mathbb{E}\left[\left|\frac{1}{\hat{f}^O_{k,g}}\right|^2\right]}\right)
\end{equation}
\hrulefill
\end{figure*}
The corresponding achievable rate for the users at the cell edge is given in \eqref{ratedloedp} on top of next page.
\begin{figure*}
\begin{equation}\label{ratedloedp}
R_{e,k}^{O}=\left(1-\frac{K}{T}\right)(1-\eta)\log_2\left(1+\frac{p_d\beta_h^k\gamma^O_{k,h}\left|\mathbb{E}\left[\frac{f_{k,h}^O}{\hat{f}_{k,h}^O}\right]\right|^2}{p_d\beta_h^k \mathrm{Var}\left[\frac{f_{k,h}^O}{\hat{f}_{k,h}^O}\right]+\mathbb{E}\left[\left|\frac{I^O_{k,h}}{\hat{f}^O_{k,h}}\right|^2\right]+\mathbb{E}\left[\left|\frac{1}{\hat{f}^O_{k,h}}\right|^2\right]}\right)
\end{equation}
\hrulefill
\end{figure*}

As in the case with perfect CSI at the users, when we set $\alpha_h^k=\gamma_{k,h}=0, \forall k$ and $\alpha_g^k=1$ we get the achievable rate of the users at the cell center in Scheme-O with $\eta=1$. Setting $\alpha_h^k=\gamma_{k,h}=0, \forall k$ and $\alpha_h^k=1$ we get the achievable rate of the users at the cell edge in Scheme-O with $\eta=0$. By using time sharing between these two extremes, we obtain all the ergodic achievable rates that Scheme-O can attain.

Table \ref{table} summarizes all the ergodic rate expressions we have obtained, they are all listed in Table \ref{table} with reference to the equation numbers.
\begin{table*}[t]
\caption{\label{table}Summary of Achievable Rate Results}
\begin{center}
    \begin{tabular}{ | c | c | c | c |}
    \hline
    Schemes (users)& Estimated CSIT, Perfect CSIR & Estimated CSIT, no CSIR & Estimated CSIT, CSIR \\ \hline
    Scheme-O (cell center)& \eqref{ratepoc} & \eqref{rateimoc} & \eqref{ratedlocdp} \\ \hline
    Scheme-O (cell edge)& \eqref{ratepoe} & \eqref{rateimoe} & \eqref{ratedloedp} \\ \hline
    Scheme-N (cell center)& \eqref{ratepnc} & \eqref{rateimnc} & \eqref{ratekgdp} \\ \hline
    Scheme-N (cell edge)& \eqref{ratepne} & \eqref{rateimne} & \eqref{ratekhdp} \\ \hline
    \end{tabular}
\end{center}
\end{table*}

Comparing the achievable rates of the different schemes under different CSI assumptions, we observe that the main difference among them is that imperfect CSI at the users is causing self-interference. Without any downlink pilots, this self-interference is proportional to the received power ($p_d\beta$), which fundamentally limits the achievable rate of the user. Therefore we can conclude that neither increasing the DL power nor putting the user closer to the BS would help much. This would not create a large SINR difference at the user, and thus we expect that Scheme-N would not provide much gain. However with DL pilots, the self-interference can be reduced substantially if we increase the DL SNR. This creates a larger SINR difference at the users and thus we expect that Scheme-N would provide higher gains.

\section{Practical Issues and Extensions}\label{practical_issues}
In this section we discuss various issues when implementing the proposed Scheme-N in practical systems and some possible extensions. Due to space limitations, these issues are discussed briefly and in-depth investigations are left for future work.

\subsection{User Pairing}
In this paper we are investigating the effects of imperfect CSI obtained through uplink training. The channels are not known a priori; the only information available at the BS regarding the channel strength is the large scale fading coefficients of the users. As a result the user pairing has to be done based on the large-scale fading coefficients $\{\beta_k\}$. This can also be observed from the achievable rate expressions. This is the same condition that has been discussed in \cite{DFP2016}. However the differences are that first, in our case there is a beamforming gain of order $M$ which effectively increases the SNR and second, the existence of self-interference caused by channel estimation errors. A detailed analysis would be interesting, but has to be left for future work.
\subsection{More than Two Users Per Group}
The proposed Scheme-N can be extended to include more than two users per group. Suppose there are $L$ users in each group $k$ and each user is labeled as user $(k,1)$ to user $(k,L)$. In the channel estimation phase they are assigned the same pilot. The BS estimates a linear combination of the channels from all $L$ users in the group. Then the BS uses this for MRT beamforming. Without loss of generality, assume they have large-scale fading coefficients ordered as $\beta_1^k \leq \beta_2^k \leq \ldots \leq \beta_L^k$. The required condition such that NOMA can be applied is that user $(k,i)$ can decode all messages intended for user $(k,j)$ for all  $j\leq i$. The condition can be written as
\begin{equation}
\mathbb{E}[\log_2(1+\mathrm{SINR}_{k,i})]\geq \mathbb{E}[\log_2(1+\mathrm{SINR}_{k,j})] \quad \forall ~i\geq j,
\end{equation}
where $\mathrm{SINR}_{k,i}$ is the effective SINR of user $(k,i)$ which has different forms according to the availability of CSI. This condition can be met by controlling the pilot power of the users. Detailed analysis of this extension is out of scope and has to be left for future work due to the limit of space.
\subsection{Users with Multiple Antennas}\label{multipleantenna}
In the case when users are equipped with more than one antenna, adding more antennas can be viewed as adding users at the same distance. Thus the same analysis and results can be applied by putting the different antennas of the same user in different groups in Scheme-N. This argument does not consider the possibility of receive beamforming at the users as it requires accurate channel estimation at the users. Since the scenario we considered is when the pilot resources are scarce, the consideration of receive beamforming at the user side is out of scope.
\subsection{Power Control}
Power control in any communication systems is crucial. We have considered both power control in the UL for the pilots and in the DL for the data. They are optimized according to the requirement of the users. In Section \ref{numerical} we will look at the rate region and a particular operating point on the Pareto boundary of the rate region which is obtained by performing power control on both UL pilots and DL data. However these are done by a grid search over different power control coefficients. More efficient algorithms for this purpose would be useful but have to be left for future work.

\color{black}
\section{Other Applications}
\subsection{Application in Multicasting}
A specific application of the techniques in Scheme-N is to multiresolution multicasting\cite{Choi2015}. In multiresoultion multicasting, signals of different resolutions are multicasted to multiple users requesting the same data. Users with low SINR decode only the low resolution signal treating the high resolution signal as noise, while users with high received SINR decode both the low and high resolution signals. It is natural to apply NOMA here since the low resolution signal is wanted by all users in the cell. In this setup we only need to use one uplink pilot for channel training and the same beamforming vector is applied to all users in the cell. This can be viewed as a special case of Scheme-N where the data intended for all users $(k,h)$ are the same and data intended for all users $(k,g)$ are the same.

\subsection{Rate-Splitting for Improving Sum Degree of Freedom}
Recently a rate-splitting approach was proposed to improve the sum degree of freedom in broadcast channels\cite{YKGY2013} which is an approach that was first used for interference channels and then called the the Han-Kobayashi scheme\cite{Gamal2012}. In the rate-splitting scheme, one selected user's message is split into a common part and a private part where the common part can be decoded by all users. The common part is super-imposed on the private part and sent with a different beamformer. All NOMA schemes can be viewed as a special case of the rate-splitting approach where there is no private part for the user $(k,h)$ and all message to user $(k,h)$ is contained in the common part. Our proposed Scheme-N can be adapted for the rate-splitting scheme to handle the problem of pilot shortage by decomposing the message of user $(k,h)$ into two parts and the analysis can be carried out using similar techniques.

\section{Numerical Results}\label{numerical}
In this section we compare the performance of the two schemes in different settings. The comparison is done by comparing the complete achievable rate regions. The achievable rate region is obtained by considering a grid of pilot power control and data power control coefficients to obtain the rate pairs for each set of power control parameters, and then take the convex hull of all the rate pairs. This assumes the use of time-sharing between different sets of power control parameters. This gives an approximate rate region which is a lower bound on the actual rate region.

\subsection{Small-Scale Antenna Systems}
The first setup that we are looking into is the case with a small number of antennas at the BS. In the simulations we choose $M=10$, $K=2$, $\beta_h=1$, $\beta_g=100$, $p_u=p_d=1$. Since we compare schemes that use the same number of pilots, we omit the pre-log penalty caused by the use of pilots for acquiring CSI. For the case without downlink CSI it has fewer pilots than the other cases.
\begin{figure}
\includegraphics[width=\linewidth]{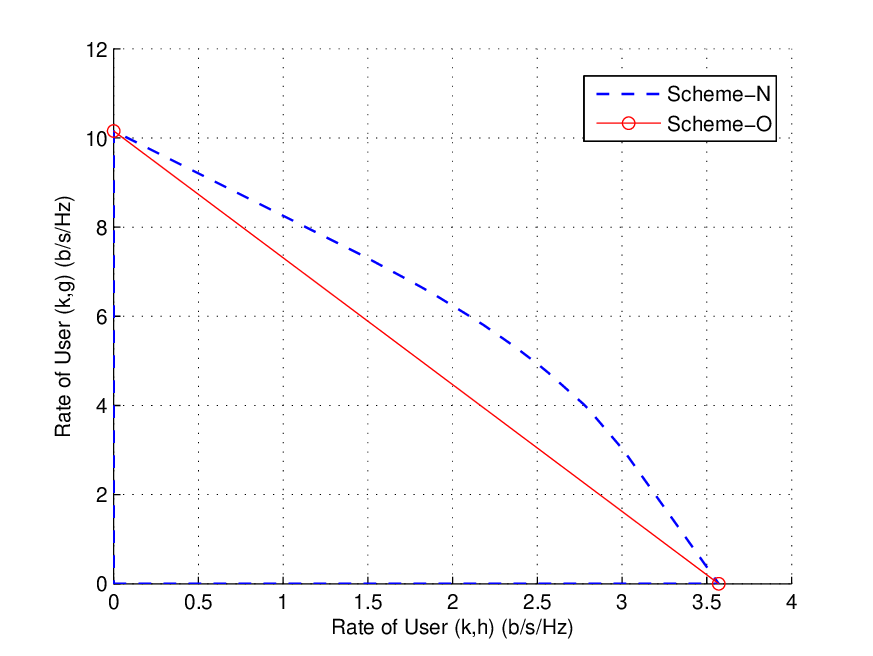}
\caption {\label{10antennasperfect} Achievable rate region with noise free uplink channel estimation and perfect CSI at the users. $M=10$, $K=2$ $\beta_h=1$, $\beta_g=100$, $p_u=p_d=1$.}
\end{figure}

Fig. \ref{10antennasperfect} shows the performance with noise free uplink estimation and perfect CSI at the users. This case represents an upper bound on the performance for practically realizable schemes. From this figure we observe that with perfect CSI available, the performance gained by using NOMA is quite significant. For example, when the rate of user $(k,h)$ is 2.5 b/s/Hz, the rate of user $(k,g)$ can be increased by almost 2 b/s/Hz.

\begin{figure}
\includegraphics[width=\linewidth]{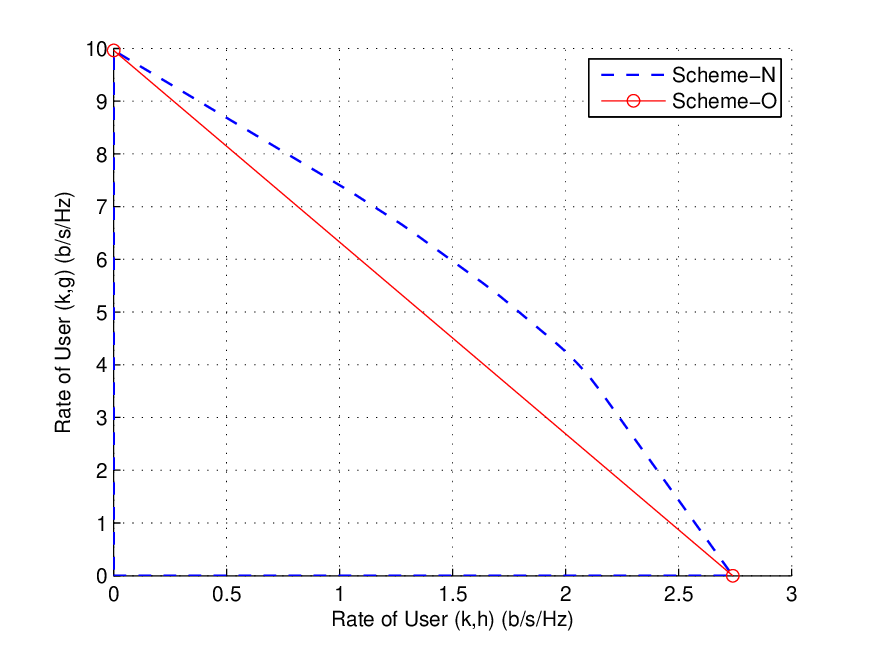}
\caption {\label{10antennasipcsit} Achievable rate region with noisy uplink channel estimation and perfect CSI at the users. $M=10$, $K=2$, $\beta_h=1$, $\beta_g=100$, $p_u=p_d=1$.}
\end{figure}

Fig. \ref{10antennasipcsit} shows the performance with noisy uplink estimation and perfect CSI at the users. Comparing with Fig. \ref{10antennasperfect} we observe that the uplink channel estimation errors do not lower the performance much for the user in the cell center. However the rate of the user at the cell edge loses more than $20\%$, due to the poor quality of the uplink channel estimate. Never the less, the gain from using NOMA is still large.

\begin{figure}
\includegraphics[width=\linewidth]{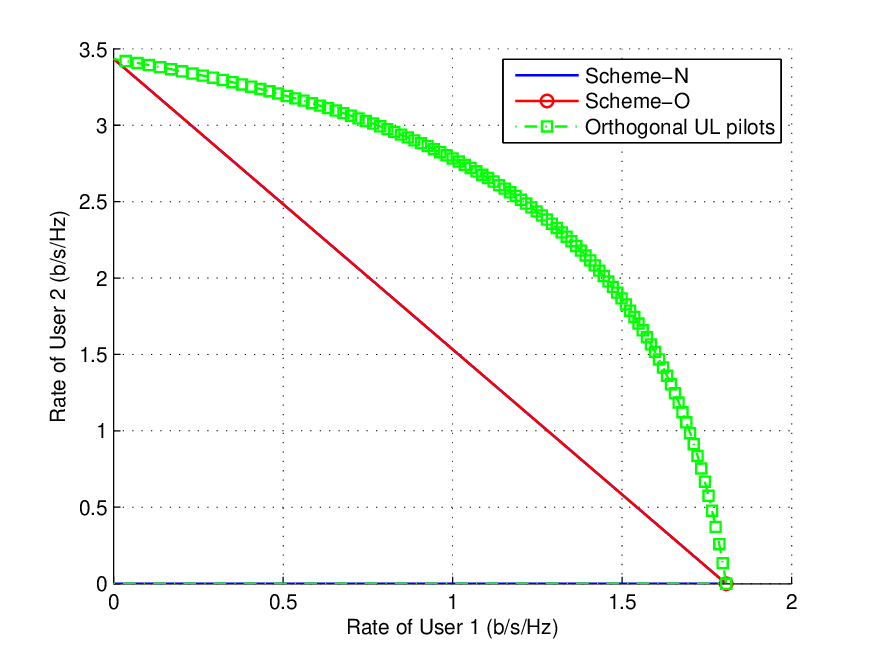}
\caption {\label{10antennasnodlpilots} Achievable rate region with noisy uplink channel estimation and no CSI at the users. $M=10$, $K=2$, $\beta_h=1$, $\beta_g=100$, $p_u=p_d=1$.}
\end{figure}

Fig. \ref{10antennasnodlpilots} shows the achievable rate region with noisy uplink estimation and no CSI at the users. Comparing to Fig. \ref{10antennasipcsit} we see that CSI at the users is critical as Scheme-N and Scheme-O are overlapping. Without CSI, Scheme-N is performing the same as Scheme-O which means there is no gain from using NOMA. We also plot the performance with orthogonal UL pilots for all $K$  users as reference. In Scheme-N we send $K/2$ uplink pilots, while with the `Orthogonal UL Pilots' scheme we send $K$ uplink pilots. In this comparison all schemes do not require downlink pilots. This shows that without taking the penalty of using more pilots, it is better to use orthogonal pilots when DL CSI is not available. When the number of pilot symbols is limited and sending $K$ orthogonal pilots is not possible, we can only compare Scheme-O and Scheme-N. There are still some gains from using NOMA with other sets of parameters (when $M$ is of the order of thousands) than the one considered in this figure, but they are marginal and applying NOMA may not be worth it since it increases the complexity and delays at the user.

\begin{figure}
\includegraphics[width=\linewidth]{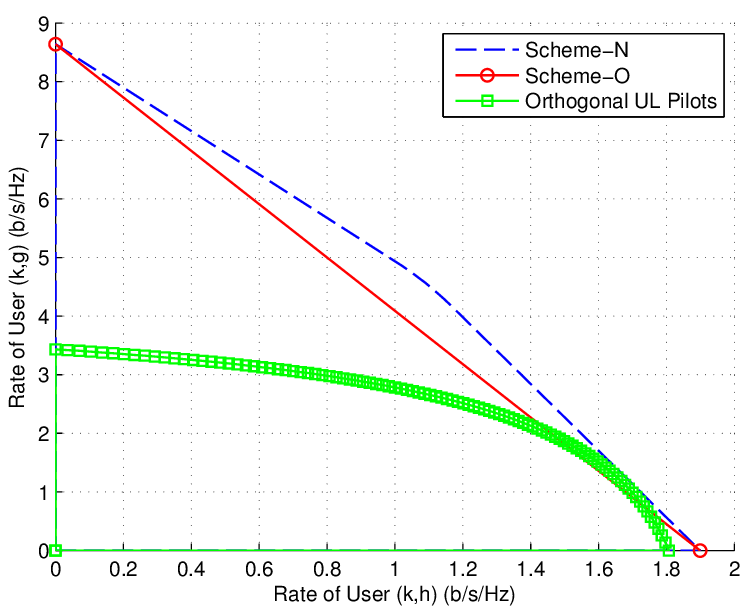}
\caption {\label{10antennas} Achievable rate region with noisy uplink channel estimation and estimated CSI at the users. $M=10$, $K=2$, $\beta_h=1$, $\beta_g=100$, $p_u=p_d=1$.}
\end{figure}

Fig. \ref{10antennas} shows the achievable rate region with noisy uplink estimation and estimated channel gains at the users, which is the most practical scenario. Comparing to Fig. \ref{10antennasnodlpilots} we see that with the estimated channel gains, we see some gains from using NOMA. We also plot the performance with orthogonal UL pilots for all $K$  users as reference. In Scheme-N we send $K/2$ uplink pilots and $K/2$ downlink pilots, while with the `Orthogonal UL Pilots' scheme we send $K$ uplink pilots and no downlink pilots. Comparing the rate regions we see that our proposed Scheme-N outperforms both traditional schemes.

\subsection{Constrained Sum Rate Comparison}
In this subsection we compare a specific operating point on the achievable rate region. We choose the point such that users at the cell edge get the same rate as in Scheme-O with $\eta=0.5$. This means that users at the cell edge do not lose any rate by using NOMA. We compare the sum rate of the whole cell under this constraint and vary the number of antennas, large-scale fading parameters. In all plots we choose $K=2$ with $1$ user at the cell edge and $1$ user in the cell center. For Scheme-O, $1$ user is scheduled in one slot, thus full power is used with $\gamma_{1,g}^O=1$ and $\gamma_{1,h}^O=1$. For Scheme-N, we vary the power between the two users to find the optimal constrained sum rate.

\begin{figure}
\includegraphics[width=\linewidth]{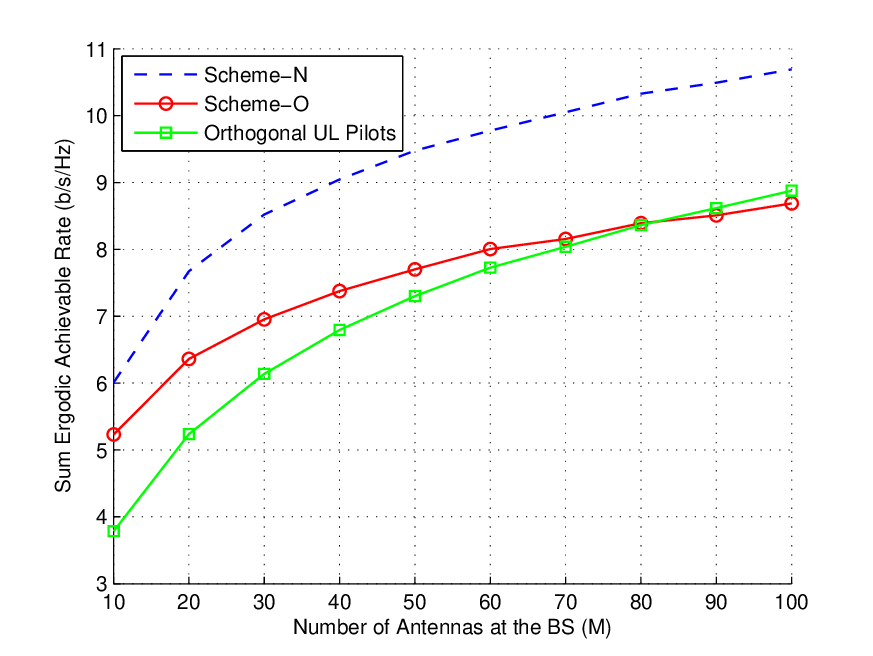}
\caption {\label{differentantennas} Sum rate with noisy uplink channel estimation and estimated channel gains at the users for different number of antennas at the BS (M). $K=2$, $\beta_h^1=1$, $\beta_g^1=100$, $p_u=p_d=1$. The rate of the user $(k,h)$ is constrained to be the rate it would get when using Scheme-O with $\eta=0.5$.}
\end{figure}

In Fig. \ref{differentantennas} we compare the constrained sum rate with different numbers of antennas $M$ at the BS with $\beta_h^1=1$, $\beta_g^1=100$, and $p_u=p_d=1$. From the plot we see that the sum rate difference between Scheme-O and Scheme-N is increasing when $M$ increases. This contradicts the common notion that NOMA is only useful when the number of antennas at the BS is less than the total number of antennas at the users\cite{DAP2016}. The reason for this is that CSI at the users is very important in NOMA, and when $M$ is small, the estimation quality is not good enough, resulting in a lower rate. When $M$ increases, the estimation quality at the users increases (due to the array gain that increases the SNR with $M$ in the DL estimation) and hence the gain from NOMA is more significant. We also observe that the performance gap between Scheme-O and the `Orthogonal UL Pilots' decreases with $M$ and eventually Scheme-O performs worse than the latter. This is due to the channel hardening effect. The more antennas at the BS, the less fluctuation in the norm of the channel vector (normalized by the number of antennas): the norm of the realization of the channel vector is almost equal to its statistical mean.

\begin{figure}
\includegraphics[width=\linewidth]{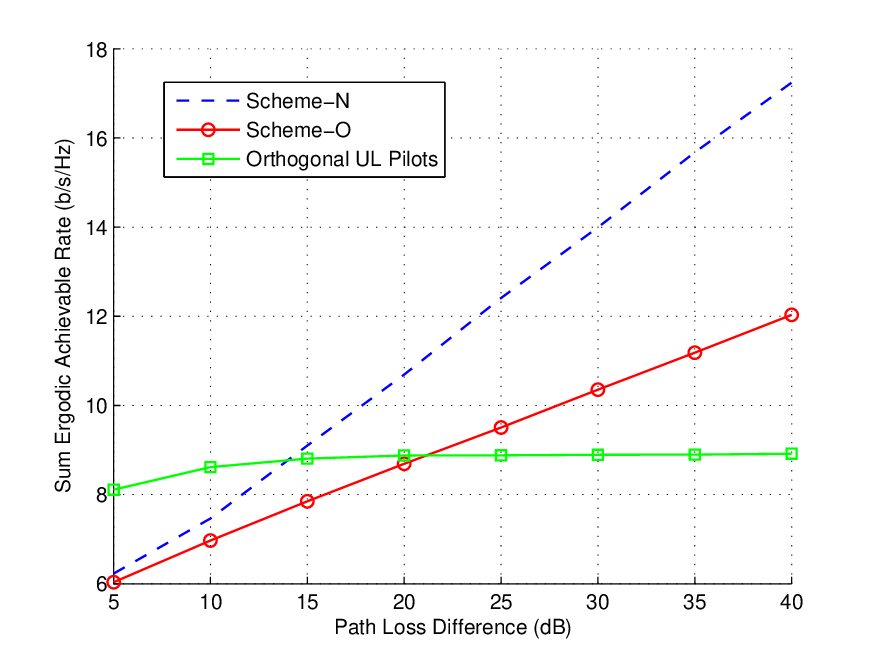}
\caption {\label{differentpathloss} Sum rate with noisy uplink channel estimation and estimated channel gains at the users with path loss differences (large-scale fading of user $(k,h)$ is fixed while large-scale fading of user $(k,g)$ is varying). $M=100$, $K=2$, $\beta_h^1=1$, $p_u=p_d=1$. The rate of the user $(k,h)$ is constrained to be the rate it would get when using Scheme-O with $\eta=0.5$.}
\end{figure}

In Fig. \ref{differentpathloss} we compare the constrained sum rate with different large-scale fading coefficients between the paired users, with $M=100$, $p_u=p_d=1$, $\beta_h^1$ is fixed to be $1$ while $\beta_g^1$ varies. From the plot we see that the sum rate difference between Scheme-O and Scheme-N is increasing with the large-scale fading difference. This is expected and matches the results for single antenna NOMA systems\cite{SKB2013}. When the large-scale fading difference is small, the orthogonal UL pilots scheme gives the best performance because both users have low SNR and therefore DL estimates are of poor quality. This verifies the importance of user pairing in NOMA.

\subsection{Effect of Number of Users or Number of Antennas at the User}
In this subsection, we look into the effect of increasing the number of users in the cell, or equivalently, the number of antennas at the users. We compare the same operating point as in the previous subsection. In the simulation we have the same number of users at the cell edge and in the cell center. The users at the cell edge have the same large-scale fading $\beta_h^k=\beta_h,~k=1, \ldots, K/2$ and the users at the cell center have the same large-scale fading $\beta_g^k=\beta_g,~k=1,\ldots,K/2$. For Scheme-O, all users that are scheduled in one slot have the same large-scale fading, thus equal power allocation with $\gamma_{k,g}^O=2/K,~k=1,\ldots,K/2$ and $\gamma_{k,h}^O=2/K,~k=1,\ldots,K/2$ is optimal in terms of achievable sum rate. For Scheme-N, we allocate equal power to each group which is also optimal for the sum rate due to the symmetry in the $K/2$ groups. The length of the coherence interval $T$ is chosen to be $200$ which is corresponding to a typical fast fading scenario.

In Fig. \ref{numberofuser}, we compare the constrained sum rates with different numbers of users, with $M=100$, $p_u=p_d=1$, $\beta_h=1$, and $\beta_g=100$. From the figure we see that Scheme-N outperforms the other schemes only when there is one group of users. As soon as there are more than one group, Scheme-N and Scheme-O are the same (Scheme-O is a special case of Scheme-N) and they are both worse than the `Orthogonal UL Pilots' scheme. This is because the inter-group interference lowers the SINR difference between the cell center user and the cell edge user, and thus NOMA does not provide any gain. This shows that the SINR difference is the key factor for NOMA to outperform the orthogonal scheme, but not the SNR difference. Moreover, from Fig. \ref{numberofuser} we also observe that when we have more users in the cell, it is better to user multiuser beamforming instead of NOMA. That is because the inter-group interference levels are the same for all schemes and that is the major factor that lowers the SINR. In this case, increasing the beamforming gain is more effective than removing the intra-group interferences.

\begin{figure}
\includegraphics[width=\linewidth]{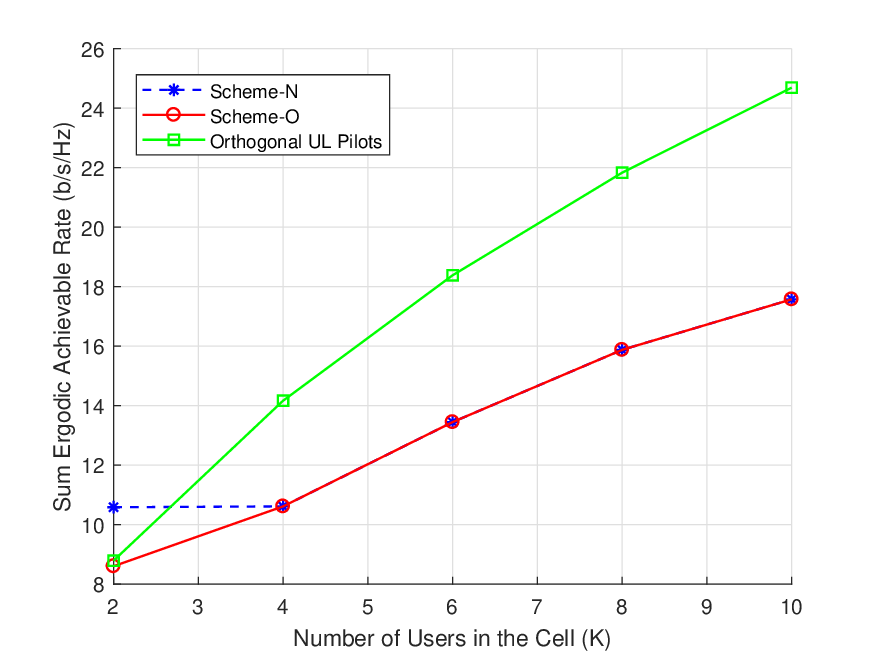}
\caption {\label{numberofuser} Sum rate with noisy uplink channel estimation and estimated channel gains at the users with different number of users $K$. $M=100$, $\beta_h=1$, $\beta_g=100$ and $p_u=p_d=1$. The rate of the user $(k,h)$ is constrained to be the rate it would get when using Scheme-O with $\eta=0.5$.}
\end{figure}
\color{black}
\section{Conclusion}
In this work we analyzed the performance of NOMA in multiuser MIMO under practical scenarios where the CSI was obtained through pilot signaling. The performance analysis was done for a conventional orthogonal scheme and a NOMA scheme under this setup. Extensive simulations were done using the derived achievable rate expressions. From the simulation results we draw the following conclusions:
\begin{enumerate}
\item NOMA works well only when high quality CSI is available at the user and there is no inter-group interference;
\item When there is more than one group, it is preferable to use multiuser beamforming instead of NOMA. In this case, we need a higher beamforming gain to enhance the SINR;
\item The gain of NOMA increases with the path loss difference between the users in the same NOMA group. When the difference is small, multiuser beamforming is preferable.
\end{enumerate}
\color{black}
The above conclusions hold when the BS precoding is restricted to MR. For other more advanced precoding methods, most importantly zero-forcing, the observations may change because accurate channel estimates are required by these methods. Some initial simulations have shown that the proposed shared-pilot scheme only provides little gain with zero-forcing precoding. More exploration is needed to find out the strategy of applying NOMA in training based systems with zero-forcing type precoding and it is left for future work.
Moreover, from our simulation results we see that CSI at the user is critical for the NOMA scheme. Instead of sending DL pilots, blind channel estimation methods designed for NOMA can help to reduce pilot overhead and therefore is worthy of exploration.
\begin{appendices}
\section{Proof of Proposition \ref{iprateh}}\label{proofiprateh}
Using results from \cite[Section 2.3.2]{MLYN2016}, we have the capacity lower bound:
\begin{equation}
R_{k,h}^{Nip}=\log_2\left(1+\frac{\left|\mathds{E}\left[c_k\sqrt{\beta_h^k}\bh_k^T\bar{\by}_{u,k}^{N*}\right]\sqrt{p_d\gamma_{k,h}}\right|^2}{\mathrm{Var}(z_{k,h})}\right).
\end{equation}

The numerator can be calculated as
\begin{equation}
\left|\mathds{E}\left[c_k\sqrt{\beta_h^k}\bh_k^T\bar{\by}_{u,k}^{N*}\right]\sqrt{p_d\gamma_{k,h}}\right|^2=p_d\lambda_{k,h}\beta_h^k\gamma_{k,h}M,
\end{equation}
and the denominator can be calculated as
\begin{equation}
\begin{aligned}
\mathrm{Var}[z_{k,h}]&=\mathrm{Var}\biggl[\biggl(c_k\sqrt{\beta_h^k}\bh_k^T\bar{\by}_{u,k}^{N*}\sqrt{p_d\gamma_{k,h}}\\
&-\mathbb{E}\left[c_k\sqrt{\beta_h^k}\bh_k^T\bar{\by}_{u,k}^{N*}\sqrt{p_d\gamma_{k,h}}\right]\biggr)s_{k,h}\\
&+c_k\sqrt{\beta_h^k}\bh_k^T\bar{\by}_{u,k}^{N*}\sqrt{p_d\gamma_{k,g}} s_{k,g}+I_{k,h}+ n_{k,h}\biggr]\\
&=\mathrm{Var}\left[c_k\sqrt{\beta_h^k}\bh_k^T\bar{\by}_{u,k}^{N*}\right]+\mathrm{Var}[I_{k,h}]+\mathrm{Var}[n_{k,h}]\\
&+\left|\mathds{E}\left[c_k\sqrt{\beta_h^k}\bh_k^T\bar{\by}_{u,k}^{N*}\right]\sqrt{p_d\gamma_{k,g}}\right|^2\\
&=p_d\beta_h^k+1+p_d\lambda_{h,k}\beta_h^k\gamma_{k,g}M.
\end{aligned}
\end{equation}

\section{Proof of Proposition \ref{iprateg}}\label{proofiprateg}
Using the results from \cite[Section 2.3.2]{MLYN2016}, we have the capacity lower bound:
\begin{equation}
R_{k,g}^{Nip}=\log_2\left(1+\frac{\left|\mathds{E}\left[c_k\sqrt{\beta_g^k}\bg_k^T\bar{\by}_{u,k}^{N*}\right]\sqrt{p_d\gamma_{k,g}}\right|^2}{\mathrm{Var}(z_{k,g})}\right).
\end{equation}

The numerator can be calculated as
\begin{equation}
\left|\mathds{E}\left[c_k\sqrt{\beta_g^k}\bg_k^T\bar{\by}_{u,k}^{N*}\right]\sqrt{p_d\gamma_{k,h}}\right|^2=p_d\lambda_{k,g}\beta_g^k\gamma_{k,g}M,
\end{equation}
and the denominator can be calculated as
\begin{equation}
\begin{aligned}
\mathrm{Var}[z_{k,g}]&=\mathrm{Var}\biggl[\biggl(c_k\sqrt{\beta_g^k}\bg_k^T\bar{\by}_{u,k}^{N*}\sqrt{p_d\gamma_{k,h}}\\
&-\mathbb{E}\left[c_k\sqrt{\beta_g^k}\bg_k^T\bar{\by}_{u,k}^{N*}\sqrt{p_d\gamma_{k,h}}\right]\biggr)s_{k,h}\\
&+\biggl(c_k\sqrt{\beta_g^k}\bg_k^T\bar{\by}_{u,k}^{N*}\sqrt{p_d\gamma_{k,g}}\\
&-\mathbb{E}\left[c_k\sqrt{\beta_g^k}\bg_k^T\bar{\by}_{u,k}^{N*}\sqrt{p_d\gamma_{k,g}}\right]\biggr)s_{k,g}+I_{k,g}+ n_{k,g}\biggr]\\
&=\mathrm{Var}\left[c_k\sqrt{\beta_g^k}\bg_k^T\bar{\by}_{u,k}^{N*}\right]+\mathrm{Var}[I_{k,g}]+\mathrm{Var}[n_{k,g}]\\
&=p_d\beta_g^k+1.
\end{aligned}
\end{equation}
\end{appendices}
\bibliographystyle{IEEEtran}
\bibliography{noma}
\end{document}